\documentclass[10pt]{article}

\usepackage{amssymb,amsmath, amsthm, setspace, color, graphics}
\usepackage{dsfont}              
\usepackage[dvips]{graphicx}
\usepackage[english]{babel}
\usepackage{float}
\usepackage{caption}
\usepackage{subcaption}
\usepackage[authoryear]{natbib}
\usepackage[pdfborder={0 0 0},final=true,colorlinks=true,linkcolor=magenta,citecolor=blue]{hyperref}
\numberwithin{equation}{section}
\theoremstyle{plain}
\newtheorem{thm}{Theorem}[section]

\newtheorem{lem}[thm]{Lemma}
\newtheorem{prop}[thm]{Proposition}
\theoremstyle{plain}
\newtheorem{assumption}[thm]{Assumption}      
\theoremstyle{definition}

\theoremstyle{definition}
\newtheorem{defn}{Definition}[section]
\theoremstyle{remark}
\newtheorem{rem}{Remark}[section]
\begin{document}
\title{\textbf{Representation Theory for Risk On Markowitz-Tversky-Kahneman Topology}\footnote{I thank Mario Ghossub for bringing my attention to related work on preference based loss aversion estimation. Research support of the Institute for Innovation and Technology Management is gratefully acknowledged.}}
\author{Godfrey Cado\u{g}an
    \thanks{Corresponding address: Institute for Innovation and Technology Management, Ted Rogers School of Management, Ryerson University, 575 Bay, Toronto, ONM5G 2C5; e-mail: \href{godfrey.cadogan@ryerson.ca}{godfrey.cadogan@ryerson.ca}; Tel: (786) 329-5469.}\\Comments welome
    }
\date{\today\vspace{-4ex}}
\maketitle
\begin{abstract}
   \noindent We introduce a representation theory for risk operations on locally compact groups in a partition of unity on a topological manifold for Markowitz-Tversky-Kahneman (MTK) reference points. We identify (1) \emph{risk torsion} induced by the flip rate for risk averse and risk seeking behaviour, and (2) a structure constant or coupling of that torsion in the paracompact manifold. The risk torsion operator extends by continuity to prudence and maxmin expected utility (MEU) operators, as well as other behavioural operators introduced by the Italian school. In our erstwhile chaotic dynamical system, induced by behavioural rotations of probability domains, the loss aversion index is an unobserved gauge transformation; and reference points are hyperbolic on the utility hypersurface  characterized by the special unitary group $SU(n)$. We identify conditions for existence of harmonic utility functions on paracompact MTK manifolds induced by transformation groups. And we use those mathematical objects to estimate: (1) loss aversion index from infinitesimal tangent vectors; and (2) value function from a classic Dirichlet problem for first exit time of Brownian motion from regular points on the boundary of MTK base topology.
\\
\\
\noindent\emph{Keywords:} representation theory, topological groups, utility hypersurface, risk torsion, chaos, loss aversion
\\ \\
\noindent\emph{JEL Classification Codes:} C62, C65, D81\\
\noindent\emph{2000 Mathematics Subject Classification:} 54H15, 37CXX
\end{abstract}
\newpage
\tableofcontents
\listoffigures
\newpage
\section{Introduction}\label{sec:Introduction}
We fill a gap in the literature on decision theory by introducing a representation theory for the Lie algebra of decision making under risk and uncertainty on locally compact groups in a topological manifold $M$. This approach is motivated by \cite[Fig.~5,~pg.~154]{Markowitz1952} who, in extending \cite{FriedmanSavage1948} utility theory, stated ``[g]enerally people avoid symmetric bets. This implies that the curve falls faster to the left of the origin than it rises to the right of the origin". In fact, \cite[pg.~155]{Markowitz1952} plainly states: ``the utility function has three inflection points. The middle inflection point is defined to be the ``customary" level of wealth. $\dotsc$ The curve is monotonially increasing but bounded; it is first concave, then convex, then concave, and finally convex". Thus, he  posited a utility function $u$ of wealth $x$ around the origin such that $u(x) > |u(-x)|$ and ``$x=0$ is customary wealth", \emph{id.,} at 155--a \emph{de facto} reference point for gains or losses in wealth. Each of the subject inflection points are critical points for risk dynamics. \cite[pg.~277]{KahnTver1979} also introduced a \emph{reference point} hypothesis. Theirs is based on ``perception and judgment", and they ``hypothesize that the value function $[v]$ for changes of wealth $[x]$ is normally concave above the reference point $(v^{\prime\prime}(x)<0,\ $ for $x >0)$ and often convex below it $(v^{\prime\prime}(x) > 0,\ x<0)$" [emphasis added], \emph{id.,} at 278. See also, \cite[pg.~303]{TverKahn1992}.

The aforementioned seminal papers support examination of risk dynamics for transformation groups in a neighbourhood of the origin [or critical points] which, by definition, are included in a topological manifold. For example, the basis sets for \cite{Markowitz1952} topology are $$U^M_\alpha = \{x|\; u(x)>|u(-x)|,\;x>0,\;-x<0<x\},\;\;\alpha\in A$$ while that for \cite{KahnTver1979, TverKahn1992} are given by $$U^{TK}_\alpha = \{x|\; u^{\prime\prime}(x)<0, x>0;\;u^{\prime\prime}(x)>0, x<0;\;-x<0<x \}\;\;\alpha\in A$$ A refined topology has basis set $U^{MTK}_\alpha=U^M_\alpha\cap U^{TK}_\alpha$. So $M\subseteq \bigcup_\alpha U^{MTK}_\alpha$ for index $\alpha\in A$. We prove that the Gauss curvature $K(\mathbf{x}_0)$ associated to a reference point $\mathbf{x}_0$ on the topological manifold of a utility hypersurface is hyperbolic, i.e. consistent with Friedman-Savage-Markowitz utility, and typically characterized by the quantum group $SU(n)$. Moreover, we introduce the concept of risk torsion and a corresponding gauge transformation for risk torsion. And extend it to the literature on \emph{prudence} spawned by \cite{Sandmo1970}. The latter typically involves \emph{precautionary savings} as a buffer against uncertain future income streams. These theoretical results provide a microfoundational bottom-up approach to results reported under rubric of decision field theory and quantum decision theory. See e.g. \cite{BusemeyerDiederich2002}; \cite{LambertMogiliansky2009}; \cite{BusemeyerPothosFrancoTrueblood2011}; \cite{YukalovSornette2010}; \cite{YukalovSornette2011}.

Other independently important results derived from our approach are value function and loss aversion index estimates. The latter being a solution to a gauge transformation for transformation groups in a Hardy space. That result is consistent with \cite[pg.~127]{KobberlingWakker2005} who argued that loss aversion is a psychological risk attribute unrelated to probability weighting and curvature of value functions in loss gain domains. Among other things, a recent paper by \cite[pg.~5]{Ghossub2012} proposed a preference based estimation procedure for loss aversion, motivated by a probability weighting operator introduced in \cite[pg.~281]{BernardGhossub2010}, and extended it to objects other than lotteries. Our estimate for loss aversion index differs from those papers because it is based on the distribution of elements of the infinitesimal tangent vector in a Lie group germ. Thus, eliminating some of the differentiability problems at the kink in \cite{KobberlingWakker2005}, and providing a preference indued loss aversion estimator. Further, we identify harmonic utility in Hardy spaces, and exploit the mean value property induced by the first exit times of Brownian motion through regular points on the boundary of a domain in MTK basis topology. Those results are summarized in Proposition \autoref{prop:ValueFunctionEstimate}.

Intuitively, our theory is based on the fact that analysis on a local utility surface extends globally if the topological manifold for that surface is paracompact. In Proposition \autoref{prop:PartionPWF} we proffer a partition of unity of probability weighting functions where each partition has a local coordinate system. The second axiom of countability and paracompactness criterion allows us to extend the analysis globally. See e.g., \cite[pp.~8-10]{Warner1983}. Furthermore, Lie group theory is based on infinitesimal generators on a topological manifold, and Lie algebras extend to linear algebra. See \cite[pg.~5]{Nathanson1979}. So our results have practical importance for analysis of behavioural data. The main results of the paper are summarized in Lemmas \autoref{lem:StructureConstant} (risk coupling), \autoref{lem:PrudenceRiskTorsion} (risk torsion), and \autoref{lem:HarmonicUtility} (harmonic utility). The rest of the paper proceeds as follows. In \autoref{subsec:BehaveOperatorProbDomain} we describe the Euclidean motions induced by risk operations. In \autoref{subsec:LieAlgebraOfRisk} we introduce the concept of risk torsion, and characterize the representation of the Lie algebra of risk. We conclude in \autoref{sec:Conclusion} with perspectives on avenues for further research.
\section{The Model}\label{sec:Model}
In this section we provide preliminaries on definitions and other pedantic used to develop the model in the sequel.
\subsection{Preliminaries}\label{subsec:Preliminaries}
\begin{defn}[Group]\cite[pp.~17-18]{Clark1971}\label{defn:Group}~\\
    A group $G$ is a set with an operation or mapping $\mu:G\times G\to G$ called a \emph{group product} which associates each ordered pair $(a,b)\in G\times G$ with an element $ab\in G$ in such a way that:
    \begin{itemize}
       \item[(1)] for any elements $a,\;b,\;c\in G$, we have $(ab)c=a(bc)$
       \item[(2)] there is a unique element $e\in G$ such that $ea=a=ae$ for any $a\in G$
       \item[(3)] for each $a\in G$ there exist $a^{-1}$, called the \emph{inverse}, such that $a^{-1}a=e=aa^{-1}$
    \end{itemize}
\end{defn}
\begin{rem}
   When $(3)$ is omitted from the definition we have a \emph{semi-group}. \hfill $\Box$
\end{rem}
\begin{figure}
   \centering
      \begin{minipage}[h]{0.4\linewidth}
         \captionof{figure}{Markowitz-Tversky-Kahneman reference point nbd}
         \label{fig:TKValueFunc}
         \centerline{\includegraphics[scale=.8]{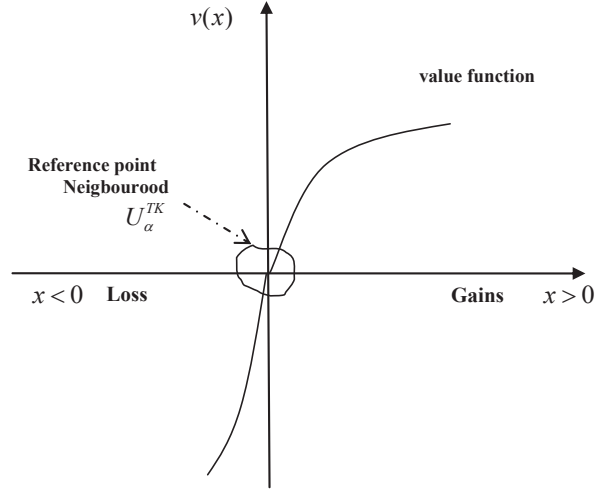}}
      \end{minipage}
   \hspace{1cm}
\end{figure}
\begin{defn}[Markowitz-Tversky-Kahneman reference point nbd topology]\label{defn:MarkowitzTverskyKahnemanTopol}
   Let $u\in C^2_0(X)$ be a real valued utility function. The reference point basis topology induced by \cite{Markowitz1952} (M) and \cite{KahnTver1979,TverKahn1992} (TK) is given by:
   \begin{description}
     \item[M]\qquad $U^M = \{x|\; u(x)>|u(-x)|,\;-x<0<x\}$
     \item[TK] \qquad $U^{TK} = \{x|\; u^{\prime\prime}(x)<0, x>0;\;u^{\prime\prime}(x)>0, x<0;\;-x<0<x \}$
     \item[MTK] \qquad $U^{MTK} = U^M\cap U^{TK}$
   \end{description}
   \hfill $\Box$
\end{defn}
\noindent A Markowitz-Tversky-Kahneman reference point neihbourood for a typical value function $v(x)$ is depicted in \autoref{fig:TKValueFunc} on page \pageref{fig:TKValueFunc}.
\begin{defn}[Compact set] See \cite[pg.~222]{Dugundji1966}\label{defn:Compact}
    A set is compact if every covering has a countable sub-cover. \hfill $\Box$
\end{defn}
\begin{defn}[Paracompact spaces]\cite[pg.~162]{Dugundji1966}\label{defn:ParacompactSpace}
   A Hausdorf space $Y$ is paracompact of each open covering of $Y$ has an open neighbourhood-finite refinement. \hfill $\Box$
\end{defn}
\begin{defn}[Topological Manifold]\cite[pg.~1]{Michor1997}\label{defn:Manifold}~\\
   A \emph{topological manifold} is a separable metrizable space $M$ which is locally homeomorphic to $\mathbf{R}^n$. So for any open neighbourhood $U$ of a point $x\in M$ there is a homeomorphism $g:U\to g(U)\subseteq \mathbf{R}^n$. The pair $(U,g)$ is called a chart on $M$. A family of charts $(U_\alpha,g_\alpha)$ such that $\cup_\alpha U_\alpha$ is a cover of $M$ is called an atlas. \hfill $\Box$
\end{defn}
\begin{rem}
   \cite[pg.~68]{Chevalley1946} provides a useful but more lengthy axiomatic definition of a manifold. \hfill $\Box$
\end{rem}
\noindent For example, $(U^M_\alpha,u_\alpha)$ and $(U^{TK}_\alpha,u_\alpha)$ are charts on some choice space manifold $M$. Whereas $\cup_\alpha U^M_\alpha$ and $\cup U^{TK}_\alpha$ are covers of $M$.
\begin{defn}[Partition of unity]\cite[pg.~8]{Warner1983}\label{defn:PartitionUnity}
   A partition of unity on $M$ is a collection $\{w_i|\;i\in I\}$ of $C^\infty$ weighting functions on $M$ such that
   \begin{description}
     \item[(a)] The collection of supports $\{\text{supp}\ w_i; i\in I\}$ is locally finite.
     \item[(b)] $\sum_{i\in I}w_i(p)=1$ for all $p\in M$, and $w_i(p)\geq 0$ for all $p\in M$ and $i\in I$.
   \end{description}
   \hfill $\Box$
\end{defn}
\begin{thm}[Existence of partition of unity on manifolds]\cite[pg.~10]{Warner1983}\label{thm:PartitionUnity}
   Let $M$ be a differentiable manifold and $\{V_\alpha,\ \alpha\in A\}$ be an open cover of $M$. Then there exists a countable partition of unity $\{w_i;\ i=1,2,\dotsc\}$, subordinate to the cover $V_\alpha$, i.e. $\text{supp}\ w_i\subset V_\alpha$, and supp$\ w_i$ compact. \hfill $\Box$
\end{thm}
\begin{rem}
   We state here that part of the theorem that pertains to paracompactness of $M$. However, it can be extended to non-compact support for $w_i$. \hfill $\Box$
\end{rem}
\noindent Theorem \autoref{thm:PartitionUnity} basically allows us extend the analysis in a reference point neighbourhood to global probability weighting functions and value function analysis. We state this formally with the following:
\begin{prop}[Partition of probability weighting functions]\label{prop:PartionPWF}
   Let $\mathbf{x}_0$ be a reference point for a real valued value function $v$ and $U_\alpha(\mathbf{x}_0)$ be a neighbourhood (nbd) of $\mathbf{x}_0$. So that $v:U_\alpha(\mathbf{x}_0)\to \mathbf{R}$. Let $\mathbf{p}_0$ be the corresponding probability attached to the reference point. Let $V_{\alpha}(\mathbf{p}_0)$ be a nbd of $\mathbf{p}_0$ for some $\alpha$. Then there exist some $C^\infty$ local probability weighting function $w_i$ with compact support, such that supp\ $w_i\subset V_\alpha$ and $0<w_i(\text{supp}\ w_i)<1$. So that $\mathbf{p}_0 \in M\Rightarrow \sum_\alpha w_\alpha(\mathbf{p}_0)=1$. \hfill $\Box$
\end{prop}
\noindent To implement Proposition \autoref{prop:PartionPWF}, we summarize the \cite[pg.~300]{TverKahn1992} topology. Let $X$ be an outcome space that includes a neutral outcome or reference point which we assign $0$. So that all other elements of $X$ are gains or losses relative to that point. An uncertain prospect is a mapping $f:\Omega\to X$ were $\Omega$ is a sample space or finite set of states of nature. Thus, $f(\omega)\in X$ is a stochastic choice. Rank $X$ in monotonic increasing order. So that a prospect $f$ is a sequence of pairs $(x_\alpha,A_\alpha)$ where $\{A_\alpha\}_{\alpha\in\mathcal{I}}$ is a discrete partition of $\Omega$ indexed by $\mathcal{I}$. In other words, the prospect $f$ is a rank ordered \emph{configuration}, i.e. sample function of a \emph{random field}, of outcomes in $X$. Let $U^{MTK}_\alpha = U^M_\alpha\cap U^{TK}_\alpha$ be a refinement of the neighbourhood topology in Definition \autoref{defn:MarkowitzTverskyKahnemanTopol}. Next, we introduce the notion of attached spaces, and proceed to apply it to the implementation at hand.
\begin{defn}[Attaching weighted probability space to outcome space]\cite[pg.~127]{Dugundji1966}.\label{defn:TopologyAttachedEventSpace}
   Let $(\Omega,\mathcal{F},P)$ be a classic probability space with sample space $\Omega$, $\sigma$-field of Borel measurable subsets of $\Omega$ given by $\mathcal{F}$, and probability measure $P$ on $\Omega$. For a sample element $\omega\in\Omega$, define $f(\omega)=x\in U^{MTK}_\alpha$ where $U^{MTK}_\alpha$ is a neighbourhood base in consequence or outcome space, and $f(\omega)$ is an act, i.e., stochastic choice. Let $U^{PWF}_\alpha$ be a $\mathcal{F}$ measurable neighbourhood base such that $P:U^{PWF}_\alpha\to U^{MTK}_\alpha$, where $P$ is a probability distribution that corresponds to $x$. Thus, $U^{PWF}_\alpha$ and $U^{MTK}_\alpha$ are two disjoint abstract spaces. Let $U^{PWF}_\alpha+ U^{MTK}_\alpha$ be the \emph{free union} of $U^{PWF}_\alpha$ and $U^{MTK}_\alpha$. Define an equivalence relation $\mathcal{R}$ by $\omega\sim (f\circ w\circ P)(\omega)$, where $w$ is a probability weighting function. The \emph{quotient space} $(U^{PWF}_\alpha + U^{MTK}_\alpha)\backslash \mathcal{R}$ is said to be $U^{PWF}_\alpha$ attached to $U^{MTK}_\alpha$ by the composite function $f\circ w\circ P$ which is written $U^{PWF}_\alpha +_{f\circ w\circ P}U^{MTK}_\alpha$. The composite function $f\circ w\circ P$ is called the attaching map. \hfill $\Box$
\end{defn}
\begin{rem}
   The interested reader is referred to \cite[\S 9]{Willard1970} for a taxonomy of examples of construction of new spaces from old in the context of quotient topology. \hfill $\Box$
\end{rem}
Let $P:A_\alpha\to U^{MTK}_\alpha(\pmb{0})$ be a mapping into a reference point neighbourhood, and $w$ be a weighting function such that $w\circ P(A_\alpha)\subseteq w(U^{MTK}_\alpha(\pmb{0}))\subseteq U^{PWF}_\alpha$, where $U^{PWF}_\alpha$ is an induced neighbourhood base cover for probability weighting assigned to uncertain events $A_\alpha$. Such mappings are permitted due to the smallness of the neighbourhoods being considered. From the outset we note that $w\in C^\infty[0,1]$ according to \cite{Prelec1998, Luce2001}.  In that way $\{U^{PWF}_\alpha\}_{\alpha\in\mathcal{I}}$ is a covering of the probabilistic manifold, i.e. we assign $w_\alpha(\mathbf{p}_0)=w(U^{PWF}_\alpha)$ so that $supp\; w_\alpha = A_\alpha$ and  $\mathbf{p}_0 \in M\Rightarrow \sum_\alpha w_\alpha(\mathbf{p}_0)=1$. For example, $M\subseteq \mathbf{R}^n\Rightarrow \mathbf{p}_0=(p^1_0,p^2_0,\dots,p^n_0)$. In other words, by Definition \autoref{defn:TopologyAttachedEventSpace}, $U^{PWF}_\alpha$ is attached to $U^{MTK}_\alpha$ by $P$ and the attached space $\large\{U^{PWF}_\alpha +_{f\circ w\circ P}U^{MTK}_\alpha\large\}_{\alpha\in I}$ is a covering of the prospect $f=(x_\alpha,A_\alpha),\;\alpha\in I$.
\begin{defn}[Lie product]\cite[pg.~105]{Guggenheim1977}\label{defn:LieProduct}~\\
   The Lie product $[\pmb{\alpha}\ \pmb{\beta}]$ of two infinitesimal vectors $\pmb{\alpha}$ and $\pmb{\beta}$ belonging to curves $\mathbf{x}(t)$ and $\mathbf{y}(t)$, respectively, is the infinitesimal vector of $(\mathbf{ab}-\mathbf{ba})(t^2)$. The substraction is understood to bee in thee sense of vector addition in $\mathbf{R}^n$. \hfill $\Box$
\end{defn}
\begin{defn}[Lie algebra]\cite[pg.~106]{Guggenheim1977}\label{defn:LieAlgebra}~\\
   The Lie algebra $\mathcal{L}(G)$ of a Lie group germ $G$ is the algebra of infinitesimal vectors defined by the Lie product. \hfill $\Box$
\end{defn}
\begin{defn}[Lie group]\cite[pg.~103]{Guggenheim1977}\label{defn:LieGroup}~\\
   A Lie group is a group which is also a differentiable manifold. A Lie group germ is a neighbourhood of the unit element $\mathbf{e}$ of a Lie group. Thus, it is possible to construct a compact Lie group from coverings of Lie group germs. Let G be a Lie group germ in a neighbourhood V of the origin \textbf{e} in $\mathbf{R}^{n}$ such that the pair of vectors is mapped (\textbf{x},\textbf{y}) $\mapsto $ \textbf{f}(\textbf{x},\textbf{y}) $\in $ $\mathbf{R}^{n}$ subject to the following axioms.
   \begin{itemize}
      \item[ (L1)] \textbf{f}(\textbf{x},\textbf{y}) is defined for all \textbf{x} $\in $ V, \textbf{y} $\in $ V

      \item [(L2)] \textbf{f}(\textbf{x},\textbf{y}) $\in $ C${}^{2}$($\mathbf{R}^n$)

      \item [(L3)] If \textbf{f}(\textbf{x},\textbf{y}) $\in $ V and \textbf{f}(\textbf{y},\textbf{z}) $\in $ V, then \textbf{f}(\textbf{f}(\textbf{x},\textbf{y}),\textbf{z}) = \textbf{f}(\textbf{x}, \textbf{f}(\textbf{y},\textbf{z}))

      \item [(L4)] \textbf{f}(\textbf{e},\textbf{y}) = \textbf{y} and \textbf{f}(\textbf{x},\textbf{e}) = \textbf{x}
\end{itemize}
   \noindent The Lie algebra $\mathcal{L}$(G) on this transformation group is given by [\textbf{a},\textbf{b}] such that
   \[\left[\alpha {\mathbf a}\ +\ \beta {\mathbf b}{\mathbf ,\ }{\mathbf c}\right]=\alpha \left[{\mathbf a},\;{\mathbf c}\right]{\rm +}\beta {\rm [}{\mathbf b},\;{\mathbf c}{\rm ]}\ \]
   \[\left[{\mathbf a}{\rm ,}\;\alpha {\mathbf b}{\rm +}\beta {\mathbf c}\right]{\rm =}\alpha \left[{\mathbf a},\;{\mathbf c}\right]{\rm +}\beta {\rm [}{\mathbf b},\;{\mathbf c}{\rm ]}\] \hfill $\Box$
\end{defn}
In the sequel, we assume that the neighbourhood $V$ which contains the Lie group germ $G$ in Definition \autoref{defn:LieGroup} is given by $V=\inf_\alpha\{U^M_\alpha\cap U^{TK}_\alpha\}$ for the topological basis in Definition \autoref{defn:MarkowitzTverskyKahnemanTopol}.
\subsection{Rotation of behavioral operator over probability domains.}\label{subsec:BehaveOperatorProbDomain}
Let $p^*$ be a fixed point probability that separates loss and gain domains. See \cite{KahnTver1979} and \cite{TverKahn1992}. Let $\mathcal{P}_{\ell } \triangleq [0,p^{*} ]$ and $\mathcal{P}_{g} \triangleq (p^{*} ,1]$ be loss and gain probability domains as indicated.  So that the entire domain is $\mathcal{P}=\mathcal{P}_\ell\cup\mathcal{P}_g$. Let $w(p)$ be a probability weighting function (PWF), and $p$ be an equivalent martingale measure.
\begin{defn}[Behavioural matrix operator]\label{defn:BehaveMatrixOperr}~\\
   The confidence index from loss to gain domain is a real valued mapping defined by the kernel function
   \begin{align}
      K:\mathcal{P}_\ell & \times\mathcal{P}_g \rightarrow [-1,1]\\
      K(p_{\ell } ,p_{g} ) &=\int _{p_{\ell } }^{p_{g} }[w (p)-p]dp=\int _{p_{\ell } }^{p_{g} }w (p)dp-\frac{1}{2} (p_{g}^{2} -p_{\ell }^{2} ),\;\;(p_\ell,p_g)\in\mathcal{P}_\ell\times\mathcal{P}_g \label{eq:BehavEconfIndex}\\
      \intertext{We note that that kernel can be transformed even further so that it is singular at the fixed point $p^*$ as follows:}
     \hat{K}(p_\ell,p_g) &= \frac{K(p_\ell,p_g)}{p_g-p_\ell}=\frac{1}{p_g-p_\ell}\int _{p_{\ell } }^{p_{g} }w (p)dp-\frac{1}{2} (p_{g} + p_{\ell } )\label{eq:SingularKernel}
\end{align}
In particular, for $\ell=1,\dotsc,m$ and $g=1,\dotsc,r$\;\;$K=[K(p_{\ell } ,p_{g} )]$ is a behavioural matrix operator. \hfill $\Box$
\end{defn}
\noindent The kernel accommodates any Lebesgue integrable PWF compared to any linear probability scheme. See e.g., \cite{Prelec1998} and \cite{Luce2001} for axioms on PWF, and \cite{Machina1982} for linear probability schemes.  Evidently, $\hat{K}$ is an averaging operator induced by $K$, and it suggests that the \emph{Newtonian potential} or \emph{logarithmic potential} on loss-gain probability domains are admissible kernels. The estimation characteristics of these kernels are outside the scope of this paper. The interested reader is referred to the exposition in \cite{Stein2010}. Let $\mathfrak{T}$ be a partially ordered index set on probability domains, and $\mathfrak{T}_\ell$ and $\mathfrak{T}_g$ be subsets of $\mathfrak{T}$ for indexed loss and indexed gain probabilities, respectively. So that
\begin{equation}
    \mathfrak{T}=\mathfrak{T}_\ell\cup\mathfrak{T}_g\label{eq:IndexProbDomain}
\end{equation}
For example, for $\ell\in\mathfrak{T}_\ell$ and $g\in\mathfrak{T}_g$ if $\ell =1,\ldots ,m;\; \; g=1,\ldots ,r$ the index $\mathfrak{T}$ gives rise to a $m\times r$ matrix operator $K=[K(p_\ell ,p_g)]$. The ``adjoint matrix`` $K^{*} =[K^{*} (p_{g} ,p_{\ell } )]=-[K(p_{\ell } ,p_{g} )]^{T} $. So $K$ transforms gain domain into loss domain--implying fear of loss, or risk aversion, for prior probability $p_\ell$. While $K^*$ is an \textit{Euclidean motion} that transforms loss domain into hope of gain from risk seeking for prior gain probability $p_g$.
\begin{defn}[Behavioural operator on loss gain probability domains]\label{defn:BehavioralOper}
   Let $K$ be a behavioral operator constructed as in \eqref{eq:BehavEconfIndex}. Then the adjoint behavioural operator is a rotation and reversal operation represented by $K^*=-K^T$. \hfill $\Box$
\end{defn}
Thus, $K^*$ captures \cite{Yaari1987} ``reversal of the roles of probabilities and payments", ie, the preference reversal phenomenon in gambles first reported by \cite{LichtensteinSlovic1973}. Moreover, $K$ and $K^*$ are generated (in part) by prior probability beliefs consistent with \cite{GilboaSchmeilder1989}. The ``axis of spin" induced by this behavioural rotation is perpendicular to the plane in which $K$ and $K^*$ operates as follows.
\begin{figure}
   \centering
      \begin{minipage}[h]{0.4\linewidth}
         \captionof{figure}{Behavioural operations on probability domains}
         \label{fig:RepTheoryRiskOperPWF}
         \centerline{\includegraphics[scale=.6]{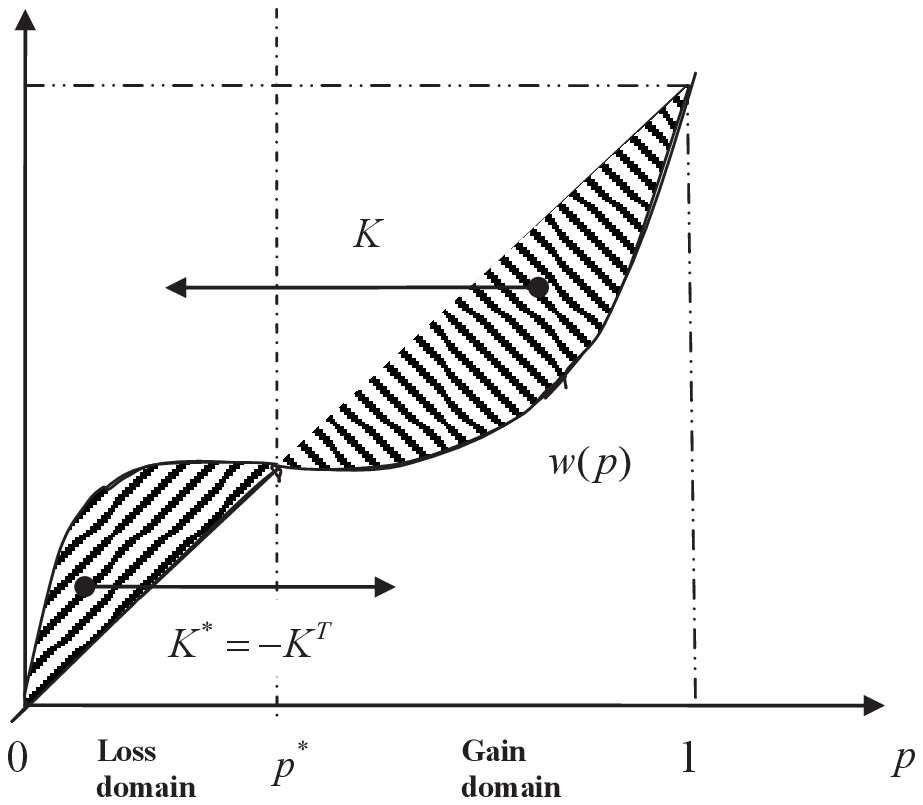}}
      \end{minipage}
   \hspace{1cm}
   \begin{minipage}[h]{0.4\linewidth}
      \centering
      \captionof{figure}{Phase portrait of behavioural orbit}
      \label{fig:TKPWFPhhaseFunc}
      \centerline{\includegraphics[scale=.6]{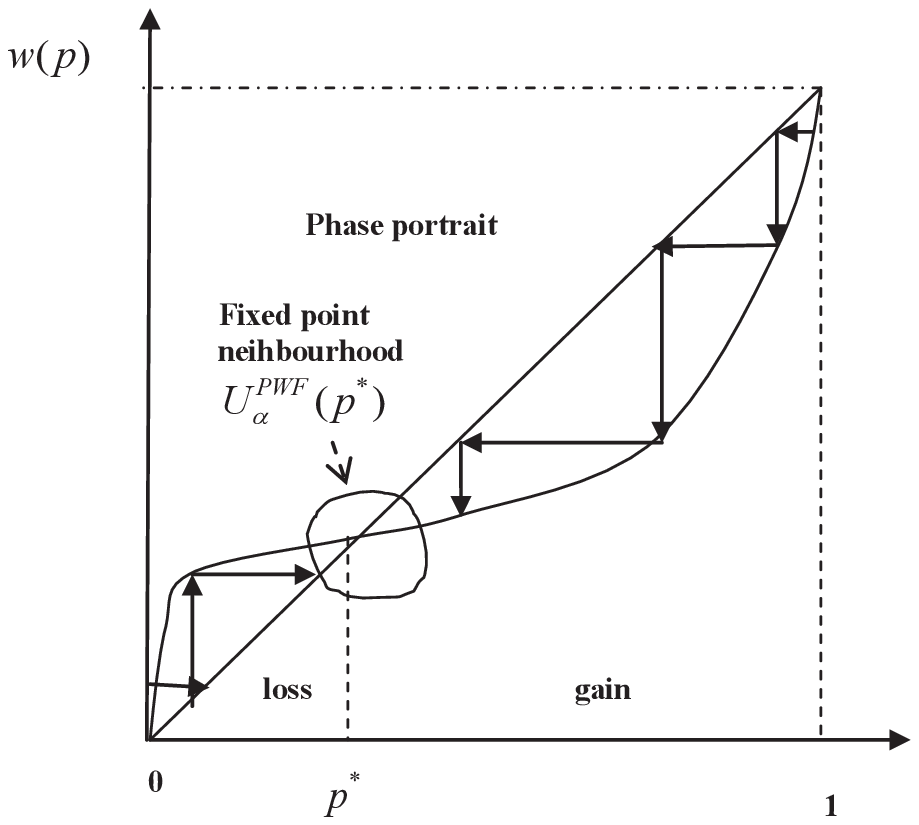}}
   \end{minipage}
\end{figure}
\subsubsection{Ergodic behaviour}\label{subsubsec:ErgodicBehave}
Consider the composite behavioural operator $T=K^T\circ K$ and its adjoint $T^*=-T^T=-T$ which is \emph{skew symmetric}.\\
\textbf{What $T^*$ does}. By definition, $T^*$ takes a vector valued function in gain domain (through $K$) that is transformed into [fear of] loss domain, and sends it back from a reduced part of loss domain (through $K^*$) where it is transformed into [hope of] gain domain. In other words, $T^*$ is a \emph{contraction mapping} of loss domain. A subject who continues to have hope of gain in the face of repeated losses in that cycle will be eventually ruined. By the same token, an operator $\widetilde{T}^*=-K\circ K^T = KK^*=-\widetilde{T}$ is a contraction mapping of gain domain. In this case, a subject who fears loss of her gains will eventually stop before she looses it all. Thus, the composite behavior of $K$ and $K^*$ is ergodic because it sends vector valued functions back and forth across loss-gain probability domains in a ``3-cycle" while reducing the respective domain in each cycle. These phenomena are depicted on page \pageref{fig:RepTheoryRiskOperPWF}. There, \autoref{fig:RepTheoryRiskOperPWF} depicts the behavioural operations that transform probability domains. \autoref{fig:TKPWFPhhaseFunc} depicts the corresponding phase portrait and a fixed point neighbourood basis set. In what follows, we introduce a behavioural ergodic theory by analyzing $T$. The analysis for $\tilde{T}$ is similar so it is omitted. Let
\begin{align}
   T=K^T\circ K &=K^TK \Rightarrow T^*=-(K^T\circ K)^T=-K^TK=K^*K=-T\\
   \intertext{Define the \emph{range} of $K$ by}
   \Delta_{K} &= \left\{g|\;Kf=g, \;\;f\in\mathfrak{D}(K)\right\}\\
   T^*f &=-K^TKf =K^*g \Rightarrow g\in\Delta_{K}\cap\mathfrak{D}(K^*)\\
   \Delta_{T^*} &= \left\{K^*g|\;g\in \Delta_{K}\cap\mathfrak{D}(K^*)\right\}\subset\mathfrak{D}(K^*)
\end{align}
Thus, $T^*$ \emph{reduces} $K^*$, i.e. it reduces the domain of $K^*$, and $T$ is \emph{skew symmetric} by construction.
\begin{lem}[Graph of confidence]\label{lem:ConfidenceGraph}~\\
   Let $\mathcal{D}(K),\;\mathcal{D}(K^*)$ be the domain of $K$, and $K^*$ respectively. Furthermore, construct the operator $T=K^*K$. We claim (i) that $T$ is a bounded linear operator, and (ii) that for $f\in\mathcal{D}(K)$ the graph $(f,Tf)$ is closed.\hfill $\Box$
\end{lem}
\begin{proof}
   See \autoref{apx:ProofLemConfGraph}
\end{proof}

\begin{prop}[Ergodic confidence]\label{prop:ErgodicConfidence}~\\
   Let $T=K^*K\;$,  $f\in\mathcal{D}(T)$ and $\mathcal{D}(K)\cap\mathcal{D}(K^*)\subseteq\mathcal{D}(T)$. Define the reduced space $\mathcal{D}(\hat{T})=\{f|\;f\in\mathcal{D}(K)\cap\mathcal{D}(K^*)\subseteq\mathcal{D}(T)$. And let $\mathfrak{B}$ be a Banach-space, i.e. normed linear space, that contains $\mathcal{D}(\hat{T})$. Let $(\mathfrak{B},\mathfrak{T},Q)$ be a probability space, such that $Q$ and $\mathfrak{T}$ is a probability measure and $\sigma$-field of Borel measureable subsets, on $\mathfrak{B}$, respectively. We claim that $Q$ is measure preserving, and that the orbit or trajectory of $\hat{T}$ induces an ergodic component of confidence.\hfill $\Box$
\end{prop}

\begin{proof}
   See \autoref{apx:ProofPropErgodicConfi}
\end{proof}
\begin{rem}
  One of the prerequisites for an ergpdic theory is the existence of a Krylov-Bogulyubov type invariant probability measure. See \cite[pg.~139]{Jost2005}. Using entropy and information, \cite[Thm.~3.2]{Cadogan2012} introduced canonical harmonic probability weighting functions with inverted S-shape in loss-gain probability domains. So that the phase portrait in Figure \autoref{fig:TKPWFPhhaseFunc} on page \pageref{fig:TKPWFPhhaseFunc}, based on an inverted S-shaped probability weighting function, is an admissible representation of the underlying chaotic behavioural dynamical system. \hfill $\Box$
\end{rem}
\begin{rem}
   Let $B$ be the set of all probabilities $p$ for which $f(p)\in\mathcal{D}(\hat{T})$. The maximal of such set $B$ is called the \emph{ergodic basin} of $Q$. See \cite[pg.~141]{Jost2005}.
   \begin{flushright}
      $\Box$
   \end{flushright}
\end{rem}
\subsubsection{Axis of spin induced by rotation} \label{subsubsection}
Let $\mathbf{x}(t)=a(t)\ \mathbf{i} + b(t)\ \mathbf{j}$ be a [vector valued]  curve in the domain $[\mathcal{D}(K)]$ of $K$ (or $[\mathcal{D}(K^*)]$ of $K^*$) with respect to a parameter $t$ such that $\mathbf{i}$ and $\mathbf{j}$ are unit vectors along the coordinate axes; and $a(t)$ and $b(t)$ be parametric curves. The ``axes of spin" for $\mathbf{x}(t)$ is perpendicular to $\mathbf{i}$ and $\mathbf{j}$. If $\mathbf{x}$ and $\mathbf{y}$ are in the same plane and inclined at an angle $\theta$ between them, then $\mathbf{x}\wedge\mathbf{y}$ is a vector perpendicular to the plane. The corresponding unit vector is given by
\begin{equation}
   \hat{\mathbf{c}}(t) = \frac{\mathbf{x}(t)\wedge\mathbf{y}(t)}{|\mathbf{x}(t)||\mathbf{y}(t)|\sin(\theta)}
\end{equation}
\begin{defn}[Spin vector] \cite[pp.~16-17]{Wardle2008}\label{defn:SpinVector}~\\
   The spin vector of \textbf{x}(t) $\in $ G, where t is a parameter, is defined as
   $$\frac{\mathbf{x}(t)\wedge\dot{\mathbf{x}}(t)}{\mathbf{x}(t)\cdot\mathbf{x}(t)}$$
   where ${\mathbf x}\left({\rm t}\right)\wedge \dot{{\mathbf x}}\left({\rm t}\right)=\ \left|{\mathbf x}\left({\rm t}\right)\right|\left|\dot{{\mathbf x}}\left({\rm t}\right)\right|sin(\theta) $, for $\theta $ the angle between \textbf{x} and $\dot{{\mathbf x}}$; and $\mathbf{x}(t).\mathbf{x}(t) =|\mathbf{x}(t)|^{2}$. \hfill $\Box$
\end{defn}
\begin{rem}
   The direction of the ``spin vector'' determines whether an agent is risk averse or risk seeking at that instant in our model. \hfill $\Box$
\end{rem}
\begin{defn}[Curvature] \cite[pg.~18]{Wardle2008}\label{defn:Curvature}~\\
   The curvature $\kappa$ is given by $$\kappa = |\pmb{t} \wedge  \pmb{t}^\prime|$$  where   $\mathbf{t}$  is the unit tangent vector relative to arc-length $s$  as parameter, and  $\pmb{t}^\prime$ is the derivative of $\pmb{t}$ with respect to $s$. In the context of a vector $\mathbf{x}(t)$ we have $$\kappa=\frac{\mathbf{x}^{\prime\prime}(t)}{[1+\mathbf{x}^\prime(t)^2]^{\tfrac{3}{2}}}$$ \hfill $\Box$
\end{defn}
\begin{defn}[Binormal]\cite[pg.~18]{Wardle2008}\label{defn:Binormal}
   The unit normal vector $\hat{\mathbf{b}}$ drawn at a point $P$ on a curve $\Gamma$ in the direction of the vector $\pmb{t} \wedge  \pmb{t}^\prime$ is called the \emph{binormal} at $P$. Specifically, $$\hat{\mathbf{b}}=\frac{\pmb{t} \wedge  \pmb{t}^\prime}{|\pmb{t} \wedge  \pmb{t}^\prime|}$$ or $$\hat{\mathbf{b}}=\frac{\mathbf{x}^\prime(t)\wedge\mathbf{x}^{\prime\prime}(t)}{\left(\frac{\mathbf{x}^{\prime\prime}(t)}{[1+\mathbf{x}^\prime(t)^2]^{\tfrac{3}{2}}}\right)}$$ \hfill $\Box$
\end{defn}
\begin{defn}[Torsion]\cite[pg.~19]{Wardle2008}\label{defn:Torsion}
   The rate of turn of the \emph{binormal} with respect to arc length $s$ at a point $P$ of a curve $\Gamma$ is called the \emph{torsion} represented by the \emph{triple scalar product}
   $$\tau=\mathbf{t}\cdot(\mathbf{t}^\prime\wedge\mathbf{t}^{\prime\prime})\kappa^2$$ which can also be written as $$\tau=\frac{(\mathbf{x}^\prime(t)\wedge\mathbf{x}^{\prime\prime}(t))\cdot\mathbf{x}^{\prime\prime\prime}(t)}{|\mathbf{x}^\prime(t)\wedge\mathbf{x}^{\prime\prime}(t)|^2}$$
   \hfill $\Box$
\end{defn}
\begin{rem}
   \cite[pg.~15]{Struik1961} defines \emph{torsion} as the rate of change of the osculating plane. The latter being the plane subtended by two consecutive tangent lines. For our purposes, torsion is roughly equal to the rate of change of Arrow-Pratt risk measure. \hfill $\Box$
\end{rem}
\noindent In \autoref{fig:RepTheoryRiskOperPWF} and \autoref{fig:TKPWFPhhaseFunc} torsion exists in a plane orthogonal to the axis of rotation induced by behavioural spin.
\section{Lie algebra of risk operators}\label{subsec:LieAlgebraOfRisk}
We define our risk operator as follows.
\begin{defn}[Logarithmic differential operator]\label{defn:LogarithmicOperator}
   A logarithmic differential operator $\ln D$ is defined for all functions $u$ in the domain $\mathfrak{D}(D)$ of $D$ such that $$(\ln Du)(x)=\text{sgn}(u^\prime(x))\ln|u^\prime(x)|,\;u^\prime(x)\neq 0$$ This definition is general enough to handle $u^\prime(x)<0$ and is undefined for $u^\prime(x)=0$. \hfill $\Box$
\end{defn}
\begin{defn}[Arrow-Pratt risk operator]\label{defn:ArrowPrattOperator}
   Let $X$ be a compact choice space, and $u\in C^2_0(X)\cap\mathfrak{D}(D)$ be a twice differentiable continuous utility function. Let $D$ be the differential operator so that $(Du)(x)=u^\prime(x)$ and $(D^2u)(x)=u^{\prime\prime}(x)$. Then the Arrow-Pratt risk operator $A$ for the risk measure $r(x)$ is given by $$r(x)=(Au)(x),\qquad\qquad A=-D\ln D=-\left(\frac{D^2}{D}\right)$$
   In the sequel we use $A_{ra}$, and $A_{rs}$ for risk averse and risk seeking operations respectively.\hfill $\Box$
\end{defn}
\noindent Let $X\subset \mathbf{R}^n$ be an open space of choice vectors, i.e., n-dimensional basket of goods; $G$ be a compact group in $X$; \textbf{x,y} $\in $ G;  and $\mathbf{u}:G\cap\mathcal{D}(K)\to V$ $\subset $ $\mathbf{R}^n$ be a vector valued utility function. By Definition \autoref{defn:Manifold}, $G$ is a topological manifold, i.e. a \emph{topological group}. Assume that $V$ is a Lie group germ induced by $G$.  For example $V$ could be a local budget set  $V\left(\mathbf{p},I\right):=\ \left\{\mathbf{x}\ \in \ \mathbf{R}^n_+\ :\ \mathbf{p}·\mathbf{x}\ \le I\right\}$ for income level $I$, price vector $ \mathbf{p} $, and consumption bundle $\mathbf{x}\in \mathbf{R}^n_+$. Let A${}_{ra}$ = $-D\ln D$ be the operator for Arrow-Pratt risk aversion (ra) described in Definition \autoref{defn:ArrowPrattOperator}. The corresponding infinitesimal vectors for $\mathbf{x},\ \mathbf{y} \in G$ are $\pmb{\alpha}=\left(\tfrac{\partial\mathbf{x}}{\partial t}\right)_{t=0}$ and $\pmb{\beta}=\left (\tfrac{\partial\mathbf{y}}{\partial t}\right )_{t=0}$, which stem from the expansion
\begin{equation}
   \mathbf{x} = \pmb{\alpha}t+\dotsc\qquad\qquad \mathbf{y}=\pmb{\beta}t+\dotsc\label{eq:VecExpansion}
\end{equation}
This gives rise to the following relationship between group operations in $G$ and vector addition of infinitesimal vectors:
\begin{thm}[Infinitesimal vectors of group product]\cite[pg.~104]{Guggenheim1977}\label{thm:InfiniVecGroupProd}~\\
   Let $\mathbf{x}, \ \mathbf{y}\in C^n(X)$ be curves in $G$, with infinitesimal vectors $\pmb{\alpha}$ and $\pmb{\beta}$. The curve $\mathbf{x}\mathbf{y}$ is differentiable and it has infinitesimal vector $\pmb{\alpha} + \pmb{\beta}$. \hfill $\Box$
\end{thm}
Second order Taylor expansion\footnote{See \cite[pp.~207-208]{TaylorMann1983}.} of ${\mathbf u}\left({\mathbf x},{\mathbf y}\right)_k$ and \eqref{eq:VecExpansion} around the origin $\mathbf{e}$ suggest that:
\begin{align}
   {\mathbf u}\left({\mathbf x},{\mathbf y}\right) &= \mathbf{u}(\mathbf{e},\mathbf{e}) + \mathbf{u}(\mathbf{x},\mathbf{e}) + \mathbf{u}(\mathbf{y},\mathbf{e}) + \left(\frac{\partial}{\partial\mathbf{x}}{\mathbf u}\left({\mathbf x},{\mathbf y}\right)+\frac{\partial}{\partial\mathbf{y}}{\mathbf u}({\mathbf x},{\mathbf y})\right)^2+\text{rem}\\
  &= (\pmb{\alpha} + \pmb{\beta})t + \frac{1}{2}\left((\pmb{\alpha} + \pmb{\beta})t\right)^2+\text{rem}\label{eq:TaylorExpandUtil}
\end{align}
\begin{align}
  \text{Let}\;\;\theta_{ij}\alpha_i\beta_j=\left(\alpha^2_i+\beta^2_j\right)\hspace{2.79in}&\label{eq:StructKonst1}\\
  \intertext{The typical element of the squared term in \eqref{eq:TaylorExpandUtil} is of the form}
  \left(\alpha^2_i+2\alpha_i\beta_j+\beta^2_j\right)= \left((2+\theta_{ij})\alpha_i\beta_j\right)\hspace{1in}&\\
  \Rightarrow\; \text{$k$-th element coefficient in vector is}\;\;a_{k.ij} = \left((2+\theta_{ij})\right)_k\hfill\qquad\qquad\qquad\qquad \label{eq:StructKonst2}&
\end{align}
So that for differentiable curves \textbf{x}(t) and \textbf{y}(t), with parameter t, i.e., one parameter group of motions, the Lie group structure for risk associated to \textbf{u}(\textbf{x},\textbf{y}), i.e., the infinitesimal generator of risk, is determined by:
\begin{align}
   &{\left({{\rm A}}_{ra}{\mathbf u}\right)}_k={\left(\left(-{\rm DlnD}\right){\mathbf u}\left({\mathbf x},{\mathbf y}\right)\right)}_k=(\pmb{\alpha}\pmb{\beta})_k\\
   &=-{\rm DlnD}\left({{\rm x}}_k\left({\rm t}\right){\rm +\ }{{\rm y}}_k\left({\rm t}\right){\rm +\ }\sum_{{\rm i,j}}{{{\rm a}}_{{\rm k.ij}}{{\rm x}}_{{\rm i}}\left({\rm t}\right){{\rm y}}_{{\rm j}}\left({\rm t}\right)}{\rm +\ }\epsilon_k\left({\mathbf x},{\mathbf y}\right)\right)\ \\
   &={\rm -}{\rm DlnD\ }\left({\rm (}{\alpha }_k+{\beta }_k{\rm )t+}\sum_{{\rm i,j}}{{{\rm a}}_{{\rm k.ij}}{\alpha }_i{\beta }_j{{\rm t}}^{{\rm 2}}}{\rm +}\epsilon_k\left({\mathbf x},{\mathbf y}\right)\right)\label{eq:Taylor2OrderExpand}
\end{align}
Here \textit{$\alpha $${}_{k}$, $\beta $${}_{k}$} are the k-th elements of the infinitesimal tangent vector $\frac{d}{dt}{\mathbf x}(t)$ and $\frac{d}{dt}{\mathbf y}(t)$, and ${{\rm a}}_{{\rm k.ij}}$ is the structure constant for second order terms in the Taylor expansion of \textbf{x}(t) and \textbf{y}(t); and $\epsilon_k\left({\mathbf x},{\mathbf y}\right)$ is $o(t^3)$\footnote{\cite[pp.~14-15]{BelinfanteKolmanSmith1966}.}.
After applying Theorem \autoref{thm:InfiniVecGroupProd}; multiplying and dividing  terms inside the brackets in \eqref{eq:Taylor2OrderExpand} by $({\alpha }_k+{\beta }_k)$, and differentiating, the differential of constant terms vanish since
\begin{equation}\label{}
  {\rm D}\ln\left({\alpha }_k+{\beta }_k\right)=0.
\end{equation}
So we can rewrite \eqref{eq:Taylor2OrderExpand} as
\begin{align}
   {\left(A_{ra}\ {\mathbf u}\right)}_k=(\pmb{\alpha}\pmb{\beta})_k &=-Dln\left[1+\left(\frac{2}{{\alpha }_k+{\beta }_k}\right)\sum_{ij}{a_{k.ij}{\alpha }_i{\beta }_jt}+\frac{\epsilon_k\left({\mathbf x},{\mathbf y}\right)}{{\alpha }_k+{\beta }_k}\right]\label{eq:APra}\\
   &\approx\ \left(\frac{{\rm -}{\rm 2}}{{\alpha }_k+{\beta }_k}\right)\sum_{{\rm ij}}{{{\rm a}}_{{\rm k.ij}}{\alpha }_i{\beta }_j}+ o(t)\\
   &=-\sum_{ij}{{\widehat{{\rm a}}}_{k.ij}}{\alpha }_i{\beta }_j+o(t)\label{eq:APra}\\
   \intertext{For risk seeking (rs), the sign of the Arrow-Pratt operator changes according to the spin vector in Definition \autoref{defn:SpinVector}. So we leave $\alpha_i\beta_j$ the same for convenience but define}
   \theta_{ji}\alpha_i\beta_j = \alpha^2_j+\beta^2_i\;&\text{and}\;a_{k.ji}=(2+\theta_{ji})_k\;\text{such that}\\
  {\left(A_{rs}{\mathbf u}\right)}_{{\rm k}} = (\pmb{\beta}\pmb{\alpha})_k &= \sum_{{\rm ij}}{{\widehat{{\rm a}}}_{{\rm k.ji}}{\alpha }_i{\beta }_j} + o(t)\label{eq:APrs}
\end{align}
 Subtract \eqref{eq:APrs} from \eqref{eq:APra} to get the $k$-th element of the Lie product vector in Definition \autoref{defn:LieProduct}
\begin{align}
   {\left(A_{ra}\ {\mathbf u}\right)}_k -{\left({{\rm A}}_{{\rm rs}}{\mathbf u}\right)}_{{\rm k}}  &= (\pmb{\alpha}\pmb{\beta})_k-(\pmb{\beta}\pmb{\alpha})_k\label{eq:LieAlgebra}\\
    &=-\sum_{{\rm ij}}{{{\widehat{\rm a}}}_{{\rm k.ij}}{\alpha }_i{\beta }_j}{\rm +\ o}\left({\rm t}\right)-\sum_{{\rm ij}}{{\widehat{{\rm a}}}_{{\rm k.ij}}{\alpha }_i{\beta }_j}{\rm +}\ o(t)\\
   \Rightarrow {(\left(A_{ra}-A_{rs}\right){\mathbf u})}_k &=-\sum_{i,j}\left(\hat{a}_{k.ij}+\hat{a}_{k.ji}\right)\alpha_i\beta_j + o(t)\\
  \Rightarrow {(\left(A_{ra}-A_{rs}\right){\mathbf u})}_k &\to \sum_{i,j}c_{k.ij}\alpha_i\beta_j\label{eq:NetRisk}
\end{align}
where the quantity
\begin{equation}
   c_{k.ij}=-\left(\hat{a}_{k.ij}+\hat{a}_{k.ji}\right)\label{eq:StructureConst}
\end{equation}
is the \emph{structure constant} for the risk operations on our topological group $G$. This gives rise to the following
\begin{defn}[Commutator]\label{defn:Commutator}~\\
   Let $\mathbf{x},\ \mathbf{y}\in G$. The \emph{commutator} of $\mathbf{x}$ and $\mathbf{y}$ is defined by $\mathbf{x}^{-1}\mathbf{y}^{-1}\mathbf{x}\mathbf{y}$. The commutator is the element that induces commutation between $\mathbf{x}$ and $\mathbf{y}$ so that $$\mathbf{x}\mathbf{y}=\mathbf{y}\mathbf{x}(\mathbf{x}^{-1}\mathbf{y}^{-1}\mathbf{x}\mathbf{y}) \qquad \Box $$
\end{defn}
\begin{defn}[Structure constant or coupling constant]
   The structure constant $c_{k.ij}$ characterizes the strength of the interaction between risk averse and risk seeking behavior. \hfill $\Box$
\end{defn}
\begin{thm}[Infinitesimal vector of commutator curve]\cite[pg.~106]{Guggenheim1977}\label{thm:InfiniVecCommutateCurve}~\\
   $[\pmb{\alpha},\pmb{\beta}]$ is the infinitesimal vector of the  commutator curve $(\mathbf{x}^{-1}\mathbf{y}^{-1}\mathbf{x}\mathbf{y})(t^2)$. \hfill $\Box$
\end{thm}
\noindent The quantities
\begin{equation}
  \widehat{a}_{k.ij}=\left(\frac{2}{\alpha_k+\beta_k}\right) a_{ k.ij}\label{eq:StructureConst2}
\end{equation}
has the following interpretation. ${\alpha }_k,\ {\beta }_k$ are the k-th element of the tangent vector $\dot{{\mathbf x}}{\mathbf (}t)$ and $\dot{\mathbf{y}}(t)$ and 2$a_{k.ij}$ is the k-th  coefficient of the second order terms which reflect the rate of spin of the tangent vectors. That is, in the context of Definition \autoref{defn:Torsion} \; $\widehat{a}_{k.ij}$ is a \emph{torsion type} constant. However, examination of \eqref{eq:APra}, \eqref{eq:APrs}  and Definition \autoref{defn:SpinVector} suggests that, in the context of our model, $\widehat{a}_{k.ij}$ reflects the rate at which agents ``flip'' between risk aversion and risk seeking in decision making. It is, in effect, \emph{risk torsion}\footnote{\cite[pg.~127]{Pratt1964} distinguished his risk measure from the curvature in Definition \autoref{defn:Curvature}. By the same token, ``risk torsion" is distinguished from the torsion in Definition \autoref{defn:Torsion}.}.
\begin{lem}[Coupling risk aversion and risk seeking torsion]\label{lem:StructureConstant}~\\
    The structure constant  $c_{k.ij}=-\left(\widehat{a}_{k.ij}+\widehat{a}_{k.ji}\right)$ associated with risk operations reflects the coupling between risk aversion and risk seeking torsion behavior in decision making. \hfill $\Box$
\end{lem}
\subsection{Prudence risk torsion}\label{subsec:PrudenceTortion}
\noindent Lemma \autoref{lem:StructureConstant} is related to the concept of \emph{prudence}, introduced by \cite{Sandmo1970} in the context of a two period model of consumption and investment, characterized by a utility function $U(C_1,C_2)$ where $C_1,\ C_2$ are consumption in periods $1$ and $2$. There, Sandmo is interested in comparing a subject's response to income and capital risk in a two period model with interest rate is $r$.
\begin{defn}[Prudence]\cite[pg.~353]{Sandmo1970}\label{defn:Prudence}
   A subject is prudent if in the face of income risk [s]he engages in \emph{precautionay savings} as a buffer against future consumption. \hfill $\Box$
\end{defn}
\cite[pg.~359]{Sandmo1970} condition for prudence rests on the relationship:
\begin{align}
   \frac{\partial}{\partial C_2}\left\{\frac{\frac{\partial^2U}{\partial C_1\partial C_2}-(1+r)\frac{\partial^2 U}{\partial C^2_2}}{\frac{\partial U}{\partial C_2}}\right\} &< 0 \label{eq:PrudenceCondition}
\end{align}
\noindent This implies the existence of $U^{\prime\prime\prime}$. In fact, \cite[pg.~354,~eq.~2]{Sandmo1970} suggests and \cite[pg.~60,~eq.~9]{Kimball1990} states that for a utility function $U\in C^3(X)$ \emph{prudence} is defined by the operation
\begin{align}
   A_pU &= -\frac{U^{\prime\prime\prime}}{U^{\prime\prime}}\\
   \intertext{which, in the context of Definition \autoref{defn:ArrowPrattOperator}, is a risk operation}
   A_{pa}U &= A_{ra}U^{\prime\prime}\label{eq:APpa}
\end{align}
\noindent  \cite[pg.~354]{Sandmo1970} described a subject's risk attitudes towards present [known] $(C_1)$ and future [uncertain] $(C_2)$ consumption thusly:
\begin{quote}
   Diagramatically it means that, starting at any point in the indifference map [for $U(C_1,C_2)$], the risk aversion function decreases with movements in the NW direction [$C_2\uparrow$] and increases with movement in the SE direction [$C_1\downarrow$]. We shall refer to this assumption as the hypothesis of \emph{decreasing temporal risk aversion}.
\end{quote}
[Emphasis added]. In the context of Lemma \autoref{lem:StructureConstant}, that description implies a coupling between the directions of the two risk operations. To see this, for some measure $\mu$ on $X$ consider the integral operator
\begin{align}
  \mathfrak{I}(\mu)= (\mathfrak{I} U)(x) &= \int_X U(x)\mu(dx),\qquad\text{so that}\\
   U = (\mathfrak{I}\circ \mathfrak{I}) U^{\prime\prime} &\Rightarrow (A_{pa}\circ \mathfrak{I}\circ \mathfrak{I}) U^{\prime\prime} = A_{ra}U^{\prime\prime}\\
   \Rightarrow A_{ra} &= (A_{pa}\circ \mathfrak{I}\circ \mathfrak{I})
\end{align}
\noindent by virtue of \eqref{eq:APpa}. We note that $\mathfrak{I}$ could be any one of several functional integration operators  characterized by a $\mu$-measure in the literature on decision making under risk and uncertainty. For example, $\mathfrak{I}$ includes but is not limited to \cite{vonNeumanMorgenstern1953}(VNM utility functional); \cite{GilboaSchmeilder1989}(maximin expected utility (MEU)); \cite{KlibanoffMarinacciMukerji2005} (smooth ambiguity); \cite{MacheroniMarinnaciRusticini2006} (variational model of that captures ambiguity); \cite{ChateauFaro2009} (operator representation of confidence preferences) or \cite{Machina1982}(local utility functional). Let $\ominus$ be the coupling action for risk averse and risk seeking prudence operations. Thus we can rewrite \eqref{eq:NetRisk} as
\begin{align}
  {(\left(A_{ra}-A_{rs}\right){\mathbf u})}_k &={\left(\left[\left(A_{pa}\ominus A_{ps}\right)\circ\mathfrak{I}\circ\mathfrak{I}\right]{\mathbf u}\right)}_k \to \sum_{i,j}c_{k.ij}\alpha_i\beta_j
\end{align}
\noindent We summarize the foregoing with the following
\begin{lem}[Prudence risk torsion]\label{lem:PrudenceRiskTorsion}
   Let $D$ be a differential operator, $A_{ra}=-D\ln D$ be Arrow-Pratt risk aversion operator, and $A_{rs}=-A_{ra}$ be the corresponding risk seeking operator. Furthermore, let $\mathfrak{I}$ be an integral operator. Define the prudence operation for risk aversion by $(A_{pa}U)=(A_{ra}\circ D\circ D)U$, assuming that the expressed functions are in the domains of the respective operators. Then the prudence risk torsion operator is given by
   $$\left(A_{ra}-A_{rs}\right)= \left[\left(A_{pa}\ominus A_{ps}\right)\circ\mathfrak{I}\circ\mathfrak{I}\right]  $$ \hfill $\Box$
\end{lem}
\subsection{Risk operator representation}\label{subsec:RiskOperRepresentation}
\noindent Perhaps most important, the risk averse operation $A_{ra}$ in \eqref{eq:APra} and risk seeking operation $A_{rs}$ in \eqref{eq:APrs} have different signs at a given point $\mathbf{x}_0\in G$.  In that case, the \emph{risk torsion operator} $A=A_{ra}-A_{rs}$  in \eqref{eq:LieAlgebra} has positive and negative eigenvalues and it belongs to the \textit{quantum group} $SU(n)$. This is a characteristic of Gauss curvature K(\textbf{x${}_{0}$}) associated with a hyperbolic point \textbf{x}${}_{0}$ on the utility hypersurface near the reference point or identity (\textbf{e}) in $G$. See \cite[pg.~213]{Guggenheim1977} and \cite[pp.~77-79]{Struik1961}. A three dimensional sketch of a hyperbolic point is depicted in \autoref{fig:TKyperbolicPt}. See \cite[pg.~83]{Struik1961}. There it can be seen that the curvature of surface area in a neighbourood of the saddle depends on the cross section or ``spin".
\begin{figure}
   \centering
      \begin{minipage}[h]{0.4\linewidth}
         \captionof{figure}{Hyperbolic point on hypersurface}
         \label{fig:TKyperbolicPt}
         \centerline{\includegraphics[scale=.8]{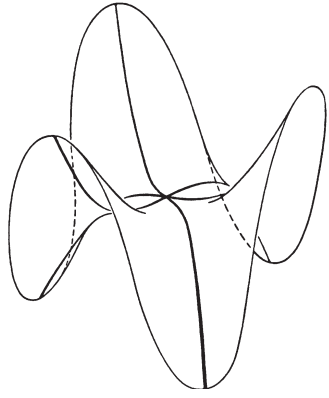}}
      \end{minipage}
   \hspace{1cm}
\end{figure}
\noindent
\begin{thm}[Lie algebra of risk operators]\label{thm:LieAlgebraRiskOper}~\\
   Let $G$ be a compact group on a differentiable manifold \textit{X} of choice vectors in $R$${}^{n}$, and $\mathbf{u}: G \times G \to G$ be a mapping of a compact group onto itself.  Let \textbf{u} be a $C^2_0\left(X\right)$ vector valued von Neuman-Morgenstern utility function defined on \textit{X}, and $\mathbf{x}(t), \mathbf{y}(t)$ be choice vectors in $G\subset X$. Define risk operators $A_{\{\cdot\}}$ such that for risk aversion A${}_{ra}$ = $-DlnD$ (risk seeking A${}_{rs}$ = $DlnD$) on the class of functions $u\in C^2_0\left(X\right)\cap \mathfrak{D}(A)$ where $\mathfrak{D}(A)$ is the domain of $A$. Then the Lie algebra $\mathcal{L}(G)$ for the risk associated to \textit{u} is the special linear group $SL_n$ of skew symmetric matrices. \hfill $\Box$
\end{thm}
\begin{thm}[Risk torsion quantum group]\label{thm:RiskTorsionQuantumGroup}
    Let \textbf{u} be a $C^3_0\left(X\right)$ vector valued von Neuman-Morgenstern utility function defined on \textit{X}, and  $A=A_{ra}-A_{rs}$ be a risk torsion operator. Then $A$ has representation in the quantum group $SU(n)$. \hfill $\Box$
\end{thm}

\noindent
\begin{assumption}See \cite[Appendix]{vonNeumanMorgenstern1953}. \label{assum:VNMutility} ~\\
   Subjects have von Neuman Morgenstern utility.
\end{assumption}
Under von Newman Morgenstern (VNM) utility framework, Arrow-Pratt risk measure is positive for risk aversion, negative for risk seeking, and the absolute value of the measure is unchanged. So $A_{rs}=-A_{ra}$. This relation has the following consequence for
\begin{align}
   X(t)\in GL_n,\qquad X(t)&=\left(x_{ik}(t)\right)_{i,k=1,\dotsc,n}
\end{align}
\begin{align}
   \intertext{By definition of group operations in $GL_n$ the identity element is the $n\times n$ matrix $I_n$, and the ``tangent matrix" is characterized by some matrix $A$ analog to \eqref{eq:VecExpansion}. We write}
   X(t) &= I_n + At + \dotsc\\
   \intertext{So that $X(0)=I_n$ is consistent with the idea that the identity element of $GL_n$ must correspond to the origin $t=0$ in accord with Definition \autoref{defn:Group}. Now}
   X(t)X(t)^T &= I_n + A^Tt + At + AA^Tt^2 +\dotsc \label{eq:GLnExpansion1}\\
   \intertext{After differentiating the left hand side and setting $t=0$, the Lie group germ structure is}
   X^\prime(0)X(0)^T &+ X(0)X^{{\prime}^T}(0) = A^T + A \label{eq:GLnExpansion2}\\
   \intertext{Let}
   A_{rs} = A\\
   \intertext{According to Definition \autoref{defn:BehavioralOper} the adjoint behavioural operator is now}
   A^*_{rs} = A_{ra} &= -A^T \label{eq:RiskProfile}\\
   \intertext{If $X(t)\in\mathcal{O}_n\subset GL_n$, i.e. $X(t)$ belongs to the group of orthogonal matrices so $X(t)X(t)^T=I_n$, then \eqref{eq:GLnExpansion1}, \eqref{eq:GLnExpansion2} and \eqref{eq:RiskProfile} reduces to}
   A^T + A = -A^*_{rs} &+ A_{rs}= A_{ra} + A_{rs}=0
\end{align}
In which case, we have \emph{skew symmetric} or \emph{antisymmetric} risk operation
\begin{equation}
   A_{rs}=-A_{ra}\label{eq:SkewSymmRisk}
\end{equation}
Thus, the risk operation in \eqref{eq:RiskProfile} is functionally equivalent to the \emph{risk torsion} operations in \eqref{eq:APra} and \eqref{eq:APrs}. This suggests that our behavioural operator $K$ in \eqref{eq:BehavEconfIndex} is well defined. More on point, the matrix representation of the skew symmetric risk operators $A_{\cdot}$ belongs to the orthonormal group $\mathcal{O}_{n}\subset SL_{n}\subset GL_{n}$ with Lie algebra $\mathcal{L}(\mathcal{O}_{n})$.
Thus, the Lie algebra $\mathcal{L}$(G) of the Lie group G is the algebra of skew symmetric\footnote{This result jibes well with \cite[pg.~268]{KahnTver1979} experiment where they reported: ``[T]he  preference  between  negative prospects is the mirror image of  the preference between positive prospects. Thus, the reflection of  prospects around 0 reverses the preference order. We label this pattern the \emph{reflection effect}''. So the risk operator is well defined.} matrices.
The foregoing is a special case of the important
\begin{thm}[Ado's theorem] \cite[pg.~202]{Nathanson1979}\label{thm:Ado}~\\
   Every finite dimensional Lie algebra $\mathcal{L}$ of characteristic zero has a finite dimensional representation. \hfill $\Box$
\end{thm}
\begin{rem}
   A field $F$ has characteristic $0$ if for any $a\in F$ and $n\in \mathbf{N}$ $\ na=0$ implies $a=0$. For example, if the ``additive identity" element of the field is $0$, it is the number of times we must add the identity to get $0$. See \cite[pg.~69]{Clark1971}. The theorem basically says, for example,  that a finite dimensional Lie algebra with characteristic $0$ has a representation in the matrix group $GL$.  \hfill $\Box$
\end{rem}
\noindent This gives rise to the following
\begin{thm}[Lie algebra of risk operation on Abelian group]~\\
    The Lie algebra $\mathcal{L}(G)$ induced by risk operations on VNM utility with support on the abelian group $G$ is that of the antisymmetric or skew symmetric matrices $\mathcal{L}(\mathcal{O}_n)$.\hfill $\Box$
\end{thm}
The infinitesimal vectors characterized by \eqref{eq:SkewSymmRisk} and our choice of $\theta_{ij}$ and $\theta_{ji}$ in \eqref{eq:StructKonst1} and \eqref{eq:StructKonst2} on page \pageref{eq:StructKonst2} require that
\begin{align}
   -\left(\widehat{a}_{k.ij}-\widehat{a}_{k.ji}\right)&=0\\
   \Rightarrow \widehat{a}_{k.ij}&=\widehat{a}_{k.ji}\Rightarrow \theta_{k.ij}=\theta_{k.ji}\\
   \Rightarrow\frac{\alpha^2_i+\beta^2_j}{\alpha_i\beta_j}&-\frac{\alpha^2_j+\beta^2_i}{\alpha_j\beta_i}=0\\
   \Rightarrow \alpha^2_j+\beta^2_i &= r^2\label{eq:InfiniVecRelation}\\
   \alpha^2_i&+\beta^2_j=r^2>0\label{eq:InfiniVecCircle}
\end{align}
So the vector elements are on a circle.

The story is different for the structure constant
\begin{align}
   c_{k.ij}=-\left(\widehat{a}_{k.ij}+\widehat{a}_{k.ji}\right)&=0\\
   \Rightarrow \widehat{a}_{k.ij}&=-\widehat{a}_{k.ji}\Rightarrow \theta_{k.ij}=-\theta_{k.ji}\\
   \Rightarrow\frac{\alpha^2_i+\beta^2_j}{\alpha_i\beta_j}&+\frac{\alpha^2_j+\beta^2_i}{\alpha_j\beta_i}=4\\
   \Rightarrow \alpha^2_j+\beta^2_i &= 4\alpha_i\beta_j-r^2\label{eq:InfiniVecRelation}\\
   \alpha^2_i&+\beta^2_j=r^2>0\label{eq:InfiniVecCircle}
\end{align}
There are two scenarios implied by \eqref{eq:InfiniVecRelation}
\begin{enumerate}
  \item[\textbf{Sc1}] If $4\alpha_i\beta_j-r^2>0$, then the vectors lie in an annulus.\label{item:VecCondit1}
  \item[\textbf{Sc2}] If $4\alpha_i\beta_j-r^2\leq 0$, then the vectors are complex valued.\label{item:VecCondit2}
\end{enumerate}
\noindent In \textbf{Sc2} above, $\alpha_i$ and or $\beta_j$ are complex valued\footnote{This subsumes the case when sgn\ $\alpha_i\beta_j\neq$sgn\ $\alpha_j\beta_i$} in the circle \eqref{eq:InfiniVecCircle}. Thus, the Hardy spaces $H^p,\;0<p<\infty$ in the unit disk $\mathbf{D}$ are admissible for our class of infinitesimal vectors where
\begin{equation}
    H^p=\left\{f\bigl|\;\underset{0<r<1}{\sup}\left(\frac{1}{2\pi}\int^{2\pi}_0|f(re^{i\theta})|d\theta\right)^{\frac{1}{p}}\right\},\;0<p<\infty
\end{equation}
Here, $f$ is a \emph{holomorphic} function, i.e. complex valued function of complex variable(s) that is complex differentiable  at every point in its domain. This class of functions also include analytic functions. So the Abelian groups implied by our structure constant are in a Hardy space. Thus, we state the following
\begin{lem}[Abelian groups supported by risk operations]\label{lem:AbelianGroupHardySpace}
   The Abelian [transformation] groups supported by risk operations are in the class of holomorphic functions in Hardy spaces $H^p$. \hfill $\Box$
\end{lem}
\begin{lem}[Harmonic utility functions]\label{lem:HarmonicUtility}
    The class of utility functions $U:X\to \mathds{C}$ that support risk operations on Abelian transformation groups in Hardy spaces is harmonic. Specifically, for infinitesimal vectors $\pmb{\alpha}$ and $\pmb{\beta}$ where $\alpha^2_i+\beta^2_j=r^2$ are vector elements, $0<r<1$, unit disk $\mathbf{D}$, and $U\in H^p(\mathbf{D})$ we have the \emph{Laplacian} $\nabla^2 U = 0$ with solution $U(re^{i\theta})=r^{|n|}e^{i\theta},\;n\in \mathds{Z}$. \hfill $\Box$
\end{lem}
According to Lemma \autoref{lem:HarmonicUtility}, for $n=1$, we have that the imaginary part $\Im{U}=u(x,y)=\sqrt{x^2+y^2}\ \sin(g(\theta))$ is an admissible harmonic utility specification for slow varying function $g(\theta),\; \theta=\tan^{-1}(\tfrac{y}{x})$. The interested reader is referred to the monograph by \cite{FollandStein1982} for ramifications of this result which is beyond the scope of this paper.
\begin{assumption} See \cite{KahnTver1979} and \cite{TverKahn1992}.\label{assum:TKutility}~\\
   Subjects have Tversky and Kahneman utility.
\end{assumption}
The story is different for prospect theory which posits a loss aversion index $\lambda$ for risk seeking over losses. So the Arrow-Pratt risk operator is asymmetrically skewed. In that case, in concert with \eqref{eq:APra} and \eqref{eq:APrs} we posit:
\begin{align}
    \widehat{a}_{k.ji} &\to\lambda \widehat{a}_{k.ij}\\
    c_{k.ij}&=-(\widehat{a}_{k.ij}+\lambda\widehat{a}_{k.ji})=0\label{eq:TKantiCommute}\\
    \Rightarrow \theta_{ik}&+\lambda\theta_{ij}+ 2(1+\lambda)=0\\
    \Rightarrow \alpha^2_j + \beta^2_i &= \frac{2\alpha_1\beta_j(1+\lambda)-r^2}{\lambda},\;\;\lambda\neq 0\label{eq:GaugeFunction}
\end{align}
So the commutator in \eqref{eq:TKantiCommute} is inflated by $\lambda$. In this case $\lambda$ is a \emph{gauge transformation} in \eqref{eq:GaugeFunction} because it has no effect on the commutativity of the underlying vectors.
\begin{defn}[Loss aversion gauge]\cite[pg.~125]{KobberlingWakker2005}\label{defn:LossAversionGauge}
   Loss aversion is a psychological gauge transformation which governs the rate of exchange between gain and loss units. \hfill $\Box$
\end{defn}
This implies that our vectors lie in Hardy spaces when $\lambda < \tfrac{r^2}{2\alpha_i\beta_j}-1, \lambda\neq 0$ and in an annulus or torus otherwise. The foregoing gives rise to the following estimates for loss aversion, and \cite{TverKahn1992} value function.
\subsection{Estimates of loss aversion index and value function}\label{subsec:EstimatesLossAverValueFunc}
\begin{prop}[Estimate of loss aversion index]\label{prop:LossAversionIndexEstimate}
   Let $\chi$ be an indicator function, $U^{MTK}_\alpha$ be a reference point ndb in Definition \autoref{defn:MarkowitzTverskyKahnemanTopol}. Let $\mathbf{D}$ be a unit disk, and $\pmb{\alpha},\pmb{\beta}$ be infinitesimal vectors in the Hardy space $H^p(\mathbf{D})\cap U^{MTK}_\alpha$ such that $\alpha^2_i +\beta^2_j=r^2>0,\;0<r<1;\; 1\leq i,j\leq n$. Let $v(x)=\chi_{\{x>0\}}v_g(x)-\chi_{\{x<0\}}\lambda v_\ell(x)$ be a value function with components $v_\ell,\;v_g$ on loss $(\ell)$-gain $(g)$ domain with loss aversion index $\lambda$. Then the loss aversion index has estimate $$0<\lambda\leq \left(\frac{r^2}{\inf_{1\leq i,j\leq n}\left(\alpha_i\beta_j\right)}-1\right)$$ \hfill $\Box$
\end{prop}
\begin{prop}[Dirichlet estimate of value function]\label{prop:ValueFunctionEstimate}
   Let $\mathfrak{D}\subset U^{MTK}_\alpha \cap H^p(\mathbf{D})$ be a domain, and $\phi$ be a bounded function for regular points on the boundary of that domain $\partial\mathfrak{D}$. Let $v\in C^2_0(\mathfrak{D})$ such that
   \begin{enumerate}
     \item $\nabla^2v =0$ in $\mathfrak{D}$
     \item $\lim_{x\to \hat{x}} v(x) = \phi(\hat{x})$ for all $\hat{x}\in\partial\mathfrak{D}$
   \end{enumerate}
   Then $$v(x) = E^x\left[\phi(B_{\tau_\mathfrak{D}})\right]$$
   where $E^x$ is the expectation operator for Brownian motion starting at $x\in\mathfrak{D}$,  $B$ is Brownian motion with respet to some probability space $(\Omega,\mathcal{F},P)$, and the first exit time from the domain is $\tau_\mathfrak{D}=\inf\{t>0|\;B_t\notin\mathfrak{D}\}$ \hfill $\Box$
\end{prop}
\begin{proof}
   See \cite[pp.~177-178]{Oksendal2003}.
\end{proof}
The latter proposition essentially says that for Brownian motion starting at an interior point $x\in\mathfrak{D}$, our estimate of the value function is its average over the distribution of its  values at those points were it potentially first exits the boundary of $\mathfrak{D}$. The results presented here fall under rubric of \emph{potential theory}. One of the problems posed by \cite[pg.~121]{KobberlingWakker2005} is how to estimate the loss aversion index $\lambda$ at the kink at $0$ depicted in the MTK neighbourhood in \autoref{fig:TKValueFunc} on page \pageref{fig:TKValueFunc}. Proposition \autoref{prop:LossAversionIndexEstimate} provides estimates for loss aversion based on elements of a tangent vector in the Lie group germ inside a Hardy space on a disk with radius normalized to 1. This approach explains the range of loss aversion index reported in \cite[pg.~1662]{AbdelBleiPara2007}.

We claim that Proposition \autoref{prop:ValueFunctionEstimate} above also allows us to obtain an estimate for loss aversion by imposing suitable boundary value conditions while circumventing the problem of existence of differentials at $0$. For example, \cite{Ito1950} provides analytics for Brownian motion in a Lie group germ consistent wit Definition \autoref{defn:LieGroup}. Further, \cite{HarrisonShepp1981} provide analytics for skew Brownian motion $X(t,\omega)$ induced by asymmetric probabilities for initial steps of a random walk, and local time around the origin where
\begin{equation}
   X(t,\omega) = B(t,\omega) + \beta L^X_t(0,\omega)\label{eq:SkewBrownianMotion}
\end{equation}
$B(t,\omega)$ is Brownian motion, $L^X_t(0,\omega)$ is local time at the origin, $|\beta|\leq 1$, and $X$ is $\mathcal{F}^B_t$-adapted. Thus, it is possible to obtain separate estimates for the value function, on loss and gain domains, based on asymmetric probabilities and the path properties of skew Brownian motion dynamics near the origin in the domain $\mathfrak{D}\subset U^{MTK}_\alpha \cap H^p(\mathbf{D})$. See also, \cite[pg.~310]{HarrisonShepp1981}. In order not to overload the paper we do not address that proposed estimation scheme here and leave it for another day.
\section{Conclusion}\label{sec:Conclusion}
 The behavioural operators and risk torsion concepts introduced in this paper provide a foundation for behavioural chaos in dynamical systems, and a mechanism for providing estimates for loss aversion and value functions. Moreover, it suggests that loss aversion index and value function estimation can be extended to \emph{potential theory}. We believe that the research paradigms suggested here would yield fruitful results that further our understanding of the data generating process for decision making under risk and uncertainty.
\appendix
\section*{Appendix of Proofs}
\addcontentsline{toc}{section}{Appendix}
\section{Proof of Lemma \autoref{lem:ConfidenceGraph}}\label{apx:ProofLemConfGraph}
\begin{proof}
~\
   \begin{itemize}
      \item[(i).] That $T$ is a bounded operator follows from the facts that the fixed point $p^*$ induces singularity in $K$ and $K^*$. Also, by construction $T$ is a contraction mapping so it is bounded. The analytic proof of those facts suggests that we let $\{T_n\}^\infty_{n=1}$ be a sequence of operators induced by an appropriate corresponding sequence of $K_n$ and $K^*_n$, and $\sigma(T_n)$--the spectrum of $T_n$. Thus, we write $\|T_n\|=\prod^{N_n}_{j=1}\lambda_j,\;\lambda_j\in\sigma(T_n)$, where $N_n=\dim\sigma(T_n)$.  Singularity implies $\lim_{N_n\rightarrow\infty}\lambda_{N_n}=0$ and for $\lambda\in\sigma(T)$, we have $\lim_{n\rightarrow\infty} \|T_n-T\|\leq \lim_{n\rightarrow\infty} |\lambda_{n}-\lambda|\|f\|=0$. Thus, $T_n\rightarrow T$ is bounded.

      \item[(ii).] Let $f\in\mathcal{D}(K)$ and $C(x)=(Kf)(x)$. So $(Tf)(x)=(K^*Kf)(x)=(K^*f^*)(x)=C^*(x)$ for $f\in\mathcal{D}(T)$. For that operation to be meaningful we must have $f^*\in\mathcal{D}(K^*)$. But $T^*=-T^T=K^TK=-T\Rightarrow f^*\in\mathcal{D}(T)$.
          According to the Open Mapping Closed Graph Theorem, see \cite[pg.~73]{Yosida1980}, the boundedness of $T$ guarantees that the graph $(f,Tf)\in\mathcal{D}(T)\times\mathcal{D}(T^*)$ is closed.
   \end{itemize}
\end{proof}
\section{Proof of Proposition \autoref{prop:ErgodicConfidence}}\label{apx:ProofPropErgodicConfi}
\begin{proof}
   Let $f\in\mathcal{D}(\widehat{T})$. Then $(\widehat{T}f)(x)=(K^*Kf)(x)=(K^*f^*)(x)=C^*$ for $f^*\in\mathcal{D}(T^*)$. But $T^*=-T^T=-(-K^TK)=K^TK=-T\Rightarrow f^*\in\mathcal{D}(T)$. Since $f$ is arbitrary, then by our reduced space hypothesis, $\widehat{T}$ maps arbitrary points $f$ in its domain back into that domain. So that $\widehat{T}:\mathcal{D}(\widehat{T})\rightarrow \mathcal{D}(\widehat{T})$. Whereupon from our probability space on Banach space hypothesis, for some measureable set  $A\in\mathcal{D}(\widehat{T})\cap\mathfrak{T}$ we have the set function $\widehat{T}(A)=A\Rightarrow \widehat{T}^{-1}(A)=A$ and $Q(\widehat{T}^{-1}(A))=Q(A)$. In which case $\widehat{T}$ is measure preserving. Now by Lemma \autoref{lem:ConfidenceGraph}, $(\widehat{T}C^*)(x)=\widehat{T}(\widehat{T}f)(x)=(\widehat{T}^2f)(x)\Rightarrow (f,\widehat{T}^2f)$ is a closed graph on $\mathcal{D}(\widehat{T})\times\mathcal{D}(\widehat{T}^*)$. By the method of induction, $(f,\widehat{T}^jf),\;j=1,2,\dotsc$ is also a graph for each $j$. In which case the evolution of the graphs $(f,\widehat{T}^jf),\;j=1,2,\dotsc$ is a dynamical system, see \cite[pg.~2]{Devaney1989}, that traces the trajectory or orbit of $f$. Now we construct a sum of N graphs and take their ``time average" to get
   \begin{align}
      f^*_N(x) &= \frac{1}{N}\sum^N_{j=1}(\widehat{T}^jf)(x)\\
      \intertext{According to Birchoff-Khinchin Ergodic Theorem, \cite[pg.~127]{GikhmanSkorokhod1969}, since $Q$ is measure preserving on $\mathfrak{T}$, we have}
      \lim_{N\rightarrow\infty}f^*_N(x) &=\lim_{N\rightarrow}\frac{1}{N}\sum^N_{j=1}(\widehat{T}^jf)(x)=f^*(x)\;\;a.s.\;Q\\
      \intertext{Furthermore, $f^*$ is $\widehat{T}$-invariant and $Q$ integrable, i.e.}
      (\widehat{T}f^*)(x) &= f^*(x)\\
      E[f^*(x)] &= \int f^*(x)dQ(x)=\lim_{N\rightarrow}\frac{1}{N}\sum^N_{j=1}\int(\widehat{T}^jf)(x)dQ(x)\\
      &= \lim_{N\rightarrow}\frac{1}{N}\sum^N_{j=1}E[({\widehat{T}}^jf)(x)]\label{eq:TimeAverage}\\
      \intertext{Moreover,}
      E[f^*(x)] &= E[f(x)]\Rightarrow (\widehat{T}E[f^*(x)])=\widehat{T}E[f(x)]=E[(\widehat{T}f)(x)]=E[C(x)]\label{eq:SpaceAverage}
   \end{align}
So the ``time average" in \eqref{eq:TimeAverage} is equal to the ``space average" in \eqref{eq:SpaceAverage}. Whence $f\in\mathcal{D}(\widehat{T})$ induces an ergodic component of confidence $C(x)$.
\end{proof}
\newpage
\bibliographystyle{chicago}
\singlespace
\addcontentsline{toc}{section}{References}
\bibliography{MathFoundationRisk}         

\end{document}